\documentclass[%
 reprint,
nofootinbib,
nobibnotes,
 amsmath,amssymb,
 aps,
pra,
notitlepage,
onecolumn
]{revtex4-1}

\usepackage{algorithm}
\usepackage{algorithmic}

\usepackage{graphicx}
\usepackage{dcolumn}
\usepackage{bm}
\usepackage{amsthm}
\usepackage{amssymb,amsfonts,amsmath,latexsym,thmtools}
\usepackage{epstopdf,refcount,nicefrac}
\usepackage{enumitem}  
\usepackage{mathtools}
\usepackage[table]{xcolor}

\usepackage{proof}
\usepackage{epsfig}
\usepackage{euscript}
\usepackage{tikz,ifdraft}
\usetikzlibrary{quantikz}
\usetikzlibrary{arrows,positioning,matrix,backgrounds,calc,fit}
\tikzcdset{
every matrix/.style={ampersand replacement=\&}
}

\usepackage{url, appendix, braket,array}

\newcolumntype{P}[1]{>{\centering\arraybackslash}m{#1}}

\newtheorem{thm}{Theorem}
\newtheorem{defn}{Definition}

\newtheorem{lemma}{Lemma}
\newtheorem{conj}{Conjecture}
\theoremstyle{remark}
\newtheorem*{remark}{Remark}

\newcommand{\comment}[1]{}
\newcommand{\imi}{\mathfrak{i}}

\newcommand{\cc}{\mathcal{C}}
\newcommand{\cd}{\mathcal{D}}

\newcommand{\ch}{\mathcal{H}}

\newcommand{\cu}{\mathcal{U}}
\newcommand{\C}{\mathbb{C}}
\newcommand{\F}{\mathbb{F}}
\newcommand{\Z}{\mathbb{Z}}
\newcommand{\I}{\mathbb{I}}
\newcommand{\R}{\mathbb{R}}

\newcommand{\crr}{\mathcal{R}}
\newcommand{\semc}{\mathcal{SC}}
\newcommand{\He}{\mathsf{H}}

\definecolor{ngreen}{rgb}{0.2,0.7,0.2}
\definecolor{nred}{rgb}{0.9,0.1,0}
\definecolor{nblue}{rgb}{0.1,0.2,0.8}

\begin{document}

\title{Efficient quantum gate teleportation in higher dimensions}
\author{Nadish de Silva}
\email{nadish.desilva@utoronto.ca}
\affiliation{Centre for Quantum Information and Foundations, \\Department of Applied Mathematics and Theoretical Physics \\University of Cambridge,  Cambridge, UK}

\begin{abstract} 
The Clifford hierarchy is a nested sequence of sets of quantum gates  critical to achieving fault-tolerant quantum computation.  Diagonal gates of the Clifford hierarchy and `nearly diagonal' semi-Clifford gates are particularly important: they admit efficient gate teleportation protocols that implement these gates with fewer ancillary quantum resources such as magic states.  Despite the practical importance of these sets of gates, many questions about their structure remain open; this is especially true in the higher-dimensional qudit setting.

Our contribution is to leverage the discrete Stone-von Neumann theorem and the symplectic formalism of qudit stabiliser mechanics towards extending results of Zeng-Cheng-Chuang (2008) and Beigi-Shor (2010) to higher dimensions in a uniform manner.  We further give a simple algorithm for recursively enumerating all gates of the Clifford hierarchy, a simple algorithm for recognising and diagonalising semi-Clifford gates, and a concise proof of the classification of the diagonal Clifford hierarchy gates due to Cui-Gottesman-Krishna (2016) for the single-qudit case.

We generalise the efficient gate teleportation protocols of semi-Clifford gates to the qudit setting and prove that every third level gate of one qudit (of any prime dimension) and of two qutrits can be implemented efficiently.  Numerical evidence gathered via the aforementioned algorithms supports the conjecture that higher-level gates can be implemented efficiently.
\end{abstract}

\maketitle

\section{Introduction}

Quantum computers hold great promise as tools for solving problems beyond the capabilities of existing classical devices.  However, their realisation in practice requires surmounting the challenges posed by the need for fine control over quantum systems.  Protection of highly sensitive quantum data against errors induced by environmental noise motivates the field of \emph{fault-tolerant quantum computation}.

Quantum error correction is complicated by the fact that the no-cloning theorem forbids the copying of quantum data and thereby forbids naive redundancy-based schemes.  Most common schemes for fault-tolerant quantum computation are based on stabiliser codes \citep{gottthesis}: generalisations of classical error-correcting codes that employ Pauli spin operations (Definition \ref{paulidefn}) and the related set of Clifford gates (Definition \ref{cliffdefn}) to protect the integrity of quantum data during computations.  The set of Clifford gates, however, do not form a sufficiently rich set of gates to perform an arbitrary quantum computation.  To achieve quantum universality, these schemes must be supplemented with the ability to perform non-Clifford gates.  

Gottesman-Chuang \cite{gottesman1999} introduced the technique of \emph{quantum gate teleportation}, a variation on standard quantum teleportation, that enables performing certain non-Clifford gates on an input state once given access to an appropriate \emph{magic state}.  Magic states are quantum resources that can be prepared in advance of a computation.  Gottesman-Chuang also introduced the \emph{Clifford hierarchy}: sets of gates that admit fault-tolerant gate teleportation protocols.  For a fixed number $n$ of qubits, the Clifford hierarchy (Definition \ref{ch}) forms a nested sequence $\cc_k^n$ of sets of $n$-qubit gates; the positive integers $k$ denotes the \emph{level} of the Clifford hierarchy.  The Pauli and Clifford gates form the first and second levels respectively.  Gates of the third level can be implemented with magic states of $2n$ qubits.  Gates of higher levels can be implemented via a recursive procedure.

With the goal of reducing the resource overhead costs of fault-tolerant quantum computation, Zhou-Leung-Chuang \cite{zhou2000methodology} introduced the one-bit teleportation protocol that enables the implementation of certain Clifford hierarchy gates, e.g. diagonal Clifford hierarchy gates, using half the ancillary resources required in the original protocol.  The class of Clifford hierarchy gates admitting these efficient teleportation protocols was expanded by Zeng-Chen-Chuang \citep{zeng2008semi} to include those which are \emph{semi-Clifford gates} (Definition \ref{sc}).  That is, those which are, in a sense, diagonalisable by Clifford gates.  In a typical quantum computation, these savings are multiplied by the many times such a gate is required.

The importance of the families of Clifford and semi-Clifford gates motivated the study of their structure.  A question of particular interest is: for which pairs $(n,k)$ are all $n$-qubit, $k$-th level gates semi-Clifford?.

\begin{itemize}\itemsep0em
\item In the case of one- or two-qubit gates, all gates of the Clifford hierarchy are semi-Clifford \citep{zeng2008semi}.
\item In the case of  three-qubit gates, all gates of the third level of the Clifford hierarchy are semi-Clifford \citep{zeng2008semi}.
\item In the case of $n>2, k > 3$, there exist $n$-qubit gates that are in the $k$-th level of the Clifford hierarchy but are not semi-Clifford \citep{zeng2008semi}.
\item In the case of $n>3, k = 3$, there exist $n$-qubit gates that are in the $k$-th level of the Clifford hierarchy but are not semi-Clifford (Gottesman-Mochon, 2009).
\end{itemize}

In this work, we will extend the study of semi-Clifford gates and their efficient gate teleportation protocols to the higher-dimensional qudit setting by raising the question: for which triples $(d,n,k)$, where $d$ is a prime dimension, are all $n$-qudit, $k$-th level gates semi-Clifford?.  Cui-Gottesman-Krishna \cite{cui2017diagonal} recently characterised the diagonal gates of the qudit Clifford hierarchy.  No prior work exists on qudit semi-Clifford gates, to the best of our knowledge.  We further give an algorithm for constructing \emph{all} gates of the Clifford hierarchy.

The magic state distillation procedure for generating magic states is more efficient for qudits than in the qubit case \cite{quditmsd}.  It is not unreasonable that practically realisable quantum computers will one day be based on qudits rather than qubits.  In any case, the mathematical and algorithmic techniques, as well as conjectures we propose with supporting numerical evidence, should benefit the wider project of elucidating the complete structure of the Clifford hierarchy and the semi-Clifford gates.

The unifying theme of the results below is the application of the discrete Stone-von Neumann theorem towards studying the Clifford hierarchy.  The Stone-von Neumann theorem \cite{neumann1931eindeutigkeit,stone1930linear}  asserts the essential equivalence of all representations of the fundamental quantum commutation relations and was originally motivated by the problem of unifying the matrix and wave mechanics pictures of early quantum theory \citep{rosenberg2004}.  This generalises a technique implicitly employed by Beigi-Shor \citep{beigi2009c3} in studying case of qubit third level gates.  We argue that this perspective is useful in understanding all levels of the Clifford hierarchy and in higher dimensions.

\subsection{Summary of main results}

\begin{itemize}\itemsep0em 
\item A streamlined proof of the discrete Stone-von Neumann theorem (Lemma \ref{uniqueunitarytuple}) and a compact expression for the unitary that carries the basic Pauli gates to prescribed, admissible targets: any unitary representation of the Heisenberg group (Theorem \ref{explicitunitary}).  

\item A simple algorithm for recursively enumerating all gates of the Clifford hierarchy (Algorithm 1).

\item A simplified statement and elementary proof of Cui-Gottesman-Krishna's classification of diagonal gates of the Clifford hierarchy (Theorem \ref{cgk}) in the single-qudit case.

\item A generalisation of the efficient gate teleportation protocol of Zhou-Leung-Chuang for qubit semi-Clifford gates to the qudit case (Section \ref{quditteleport}).

\item A novel strengthening of a characterisation of semi-Clifford gates (Theorem \ref{semicliff}) and an algorithm for recognising and diagonalising semi-Clifford gates (Algorithm 2).

\item A proof that all third level gates of one-qudit (of any prime dimension) or two-qutrits are semi-Clifford (Theorem \ref{3rdlvlisSC}) and numerical evidence suggesting that this extends to other $(d,n,k)$.

\end{itemize}

\section{Mathematical background}

\subsection{Notation}

We shall denote the imaginary unit by $\imi$ to distinguish it from our use of $i$ as an indexing variable.

Suppose $d$ is prime, $n \geq 1$, and let $\omega = e^{\imi \, 2 \pi  / d}$.  The set $\{1, ..., n\}$ is denoted by $[n]$.  For $\hat{z} \in \Z_d^n$, the ket $\ket{\hat{z}} = \ket{z_n...z_1}$.

\begin{defn}
For any $n \geq 1$ and function $f: \Z^n_d \to \C$, the diagonal matrix $D[f] \in M_{d^n}(\C)$ is defined by $$D[f] \ket{\hat{z}} = f(z) \ket{\hat{z}}.$$
\end{defn}

Identity matrices are denoted by $\I$; its dimension is given by context. The set of $d^n \times d^n$ unitary complex matrices is denoted by $\cu(d^n)$.  Given $n$ unitaries $U_1, ..., U_n$ and a vector $\hat{p}$ of $n$ integers, we denote by $U^{\hat{p}}$ the product $U_1^{p_1} \cdots U_n^{p_n}$.

\subsection{The Heisenberg group and the Stone-von Neumann theorem}

The Heisenberg \emph{canonical commutation relations} are the mathematical branch point at which quantum theory diverges from classical theory.   They lead to the idea that observable quantities are no longer modelled by scalar-valued functions on phase space but by noncommuting linear operators.  Ignoring constants, with $P$ representing momentum and $Q$ representing displacement of a system with one continuous degree of freedom:  $$[P,Q] = \imi.$$

\noindent Systems with multiple degrees of freedom are represented by tuples $\{(P_i,Q_i)\}_{i \in [n]}$ satisfying $[P_i,Q_i] = \imi$ and, for $i \neq j$, $[P_i,Q_j] = [P_i,P_j] = [Q_i,Q_j] = 0$.

As noted by Weyl, this equation has no solutions with $P$ or $Q$ a bounded operator.  To sidestep this technical issue, he introduced his exponentiated form.  Let $U(s) = e^{\imi \, sP}$ and $V(t) = e^{\imi \, tQ}$ be two groups of unitaries indexed by the parameters $s,t \in \R$.  They obey:  $$U(s) V(t) = e^{\imi \,s\cdot t} \,V(t) U(s).$$  These are similarly generalised to multiple degrees of freedom: for $i \neq j$, $[U_i(s),V_j(t)] = [U_i(s),U_j(t)] = [V_i(s),V_j(t)] = 0$.

This relation is instantiated both by Heisenberg's infinite-dimensional matrices and by Schr\"odinger's multiplication and differentiation operators on $L^2(\R)$ (which act on functions $f \in L^2(\R)$ by $X(f)(x) = xf(x)$ and $P(f)(x) = \frac{\partial}{\partial x}f(x)$ respectively).  The motivation of the Stone-von Neumann theorem was to assure the equivalence of the matrix mechanics picture and the wave mechanics picture of quantum theory.  It asserts that all manifestations of the Weyl commutation relations are unitarily equivalent.  The following modern statement of the theorem is found in e.g. \cite{qtfm}.

\begin{thm}[Stone-von Neumann \cite{neumann1931eindeutigkeit,stone1930linear} , 1930]  Suppose $\{A_i\}_{i \in [n]}$, $\{B_i\}_{i \in [n]}$ are self-adjoint operators that act irreducibly on a Hilbert space $\ch$, i.e. the only closed subspaces of $\ch$ invariant under every $e^{\imi \, sA}$ and $e^{\imi \, tB}$ are trivial, and satisfy the Weyl commutation relations.  Then there exists a unitary map $G: \ch \to L^2(\R)$, unique up to phase, such that: $$Ge^{\imi \, sA_j}G^{-1} = e^{\imi \, sX_j} \quad \mathrm{ and }\quad  Ge^{\imi \, tB_j}G^{-1} = e^{\imi \, tP_j} .$$ 

\end{thm}

The Weyl relations are most conveniently encapsulated in terms of the Heisenberg group.

\begin{defn}For any $n \geq 1$ and field $\F$, the \emph{$n$-Heisenberg group} $\He_n(\F)$ is $$\left\{\begin{pmatrix}
1 & p_1 & \cdots & p_n & c\\
 & 1 & & & q_1\\
 &  & \ddots & & \vdots\\
 & & & 1 & q_n\\
 & & & & 1
\end{pmatrix} \in M_{n+2}(\F) \;|\; p_1, ..., p_n, q_1, ..., q_n ,c \in \F\right\}.$$

\end{defn}

Unitary representations of $\He_n(\R)$ with the specified central character $c \mapsto e^{\imi c}$ give $2n$ one-parameter groups of unitaries satisfying the Weyl relations.

\subsection{Higher-dimensional Pauli gates}

In quantum computing, we are typically more interested in finite-dimensional systems.  The Heisenberg group can be defined over finite fields for prime $d$.  The canonical representation of $\He_n(\Z_d)$ gives the generalised \emph{Pauli gates}.  For the $n=1$ case: $$\begin{pmatrix}
1 & p & c\\
 & 1 &  q\\
& & 1
\end{pmatrix} \mapsto \omega^c Z^{p} X^{q}$$ where, in the standard orthonormal basis $\{ \ket{z} \;|\; z \in \Z_d\}$ of $\C^d$:
$$Z \ket{z} = \omega^z \ket{z} \quad X \ket{z} = \ket{z + 1 \text{ (mod } d\text{)}}.$$

For $n>1$ and $i \in [n]$, define $Z_i \in M_{d^n}(\C)$ to be a tensor product of $n-1$ identity matrices of size $d \times d$ with $Z$ in the $i$-th factor: $\I \otimes ...  \otimes  Z \otimes ...  \otimes  \I$; $X_i$ is defined similarly.  The $Z_i, X_i$ are the \emph{basic Pauli gates} and satisfy the commutation relation $$ Z_i X_i = \omega X_i Z_i$$ with basic gates associated to different factors commuting.

These give a representation of $\He_n(\Z_d)$: $$\begin{pmatrix}
1 & p_1 & \cdots & p_n & c\\
 & 1 & & & q_1\\
 &  & \ddots & & \vdots\\
 & & & 1 & q_n\\
 & & & & 1
\end{pmatrix} \mapsto \omega^c \prod_{i \in [n]} Z_i^{p_i} X_i^{q_i}.$$

\begin{defn}\label{paulidefn}The group of \emph{Pauli gates} is denoted $\cc^n_1 = \{\omega^c Z^{\hat{p}} X^{\hat{q}}\ \; | \; c \in \Z_d, (\hat{p}, \hat{q}) \in \Z_d^{2n} \}$
\end{defn}

A discrete analogue of the Stone-von Neumann theorem, a corollary of Mackey's vast generalisation \cite{mackey1949theorem} of the original version, asserts the unitary equivalence of different representations of $\He_n(\Z_d)$ that agree on the centre.  Due to its importance to the remainder of this paper, we defer an elementary, self-contained proof to the following section (Lemma \ref{uniqueunitarytuple}).

We will later require the two following simple lemmas.

\begin{lemma}\label{xdiagcomm}
For $n \geq 1$ and a function $f: \Z^n_d \to \C$, $$X_i D[f] = D[T_if] X_i$$ where $T_if(\hat{z}) = f(z_1, ..., z_i -1, ..., z_n)$ denotes a translation in the $i$-th component of $f$.
\end{lemma}

\begin{proof}
Both sides map $\ket{\hat{z}}$ to $f(z) \ket{z_1, ..., z_i + 1, ..., z_n}$.
\end{proof}

\begin{lemma}\label{diagmatcommz}
For $n \geq 1$, a matrix $M \in M_{d^n}(\C)$ is diagonal if and only if it commutes with $Z_i$ for all $i \in [n]$.
\end{lemma}

\begin{proof}
If $M$ is diagonal, it commutes with the diagonal $Z_i$.  Conversely, if $M$ commutes with each $Z_i$, it commutes with each rank-1 projector onto a standard basis vector $\ket{\hat{z}}\bra{\hat{z}} = d^{-n} \sum_{\hat{p} \in \Z_d^n} \omega^{-\hat{p} \cdot \hat{z}} Z^{\hat{p}}$ and is therefore diagonal.
\end{proof}

\subsection{The Clifford group and the symplectic phase space formalism}

The Pauli gates form the basis of the error-correcting codes necessary for making quantum computation a practical reality.  The set of Clifford gates can be performed fault-tolerantly on data encoded using these stabiliser codes.

\begin{defn}\label{cliffdefn}The \emph{Clifford gates} are those unitaries that preserve the group of Pauli gates under conjugation:$$\cc_2^n = \{G \in \cu(d^n) \;|\; G \cc_1^n G^* \subset \cc_1^n \}.$$

\end{defn}
\noindent Being the normaliser of a subgroup of the unitaries, the set of Clifford gates form a group.

When $d$ is an odd prime, the Pauli and Clifford groups admit a rich phase space formalism in terms of a discrete symplectic vector space \cite{gross2006hudson}.  We limit our presentation of this topic to those elements required below.

The \emph{phase space} is $\Z_d^{2n}$ and a typical phase point is usually denoted $(\hat{p},\hat{q})$.  The symplectic inner product $[\cdot,\cdot]: \Z_d^{2n} \times \Z_d^{2n} \to \Z_d$ is defined by $[(\hat{p_1},\hat{q_1}),(\hat{p_2},\hat{q_2})] = \sum_{i \in [n]} p_1 q_2 - p_2 q_1$ (mod $d$).  To each phase point, we associate the Pauli gate $W(\hat{p},\hat{q}) = \omega^{-2^{-1}\hat{p}\cdot\hat{q}} Z^{\hat{p}} X^{\hat{q}} $ where $2^{-1}$ denotes the multiplicative inverse of $2$ in $\Z_d$.  They obey the multiplication law: $W(\hat{p_1},\hat{q_1})W(\hat{p_2},\hat{q_2}) = \omega^{2^{-1} [(\hat{p_1},\hat{q_1}),(\hat{p_2},\hat{q_2})]} W(\hat{p_1}+\hat{p_2},\hat{q_1}+\hat{q_2})$.

A set of $m$ Pauli gates $\{\omega^{c_i} Z^{\hat{p_i}} X^{\hat{q_i}} \}_{i \in [m]}$ is \emph{independent} if no nontrivial product of them equals the identity.  This is equivalent to the set $\{(\hat{p_i},\hat{q_i})\}_{i \in [m]}$ being a linearly independent subset of $\Z_d^{2n}$.  

The Clifford gates, up to phase, are in correspondence with affine symplectic transformations of the phase space.  First, we define the group of projective Cliffords as the quotient group of Cliffords modulo phase: $[\cc_2^n] = \cc_2^n / \mathbb{T}$.  The group $Sp(n,\Z_d) \ltimes \Z_d^{2n}$ of affine symplectic transformations of $\Z_d^{2n}$ is are pairings of $2n \times 2n$ symplectic matrices and translations in $\Z_d^{2n}$ with the composition law: $$(S,v) \circ (T,w) = (ST, Sw + v).$$  There is a (\emph{Weil} or \emph{metaplectic}) projective representation $\rho: Sp(n,\Z_d) \ltimes \Z_d^{2n} \to [\cc_2^n]$ that is an isomorphism between the groups of affine symplectic transformations and projective Cliffords. 

\subsection{The Clifford hierarchy}

While Clifford gates can be implemented fault-tolerantly, they are not a sufficiently rich gate set to perform arbitrary quantum computations.  Motivated by the need to implement non-Clifford gates fault-tolerantly, Gottesman-Chuang introduced the Clifford hierarchy.  

\begin{defn}\label{ch}The \emph{Clifford hierarchy} is an inductively defined sequence of sets of gates.  For $k > 1$, the \emph{$k$-th level of the Clifford hierarchy} is the set:  $$\cc^n_k = \{G \in \cu(d^n) \;|\; G \cc_1^n G^* \subset \cc^n_{k-1}\}.$$
\end{defn}

The levels of the Clifford hierarchy are nested: $\cc_k^n \subset \cc_{k+1}^n$.  While the first two levels form groups, higher levels do not.  However the sets $\{G \in \cc^n_k \, | \, G\text{ is diagonal}\}$ do form groups.  The sets $\cc^n_k$ are closed under left or right multiplication by Clifford gates: for $k > 1$, $\cc_2^n \, \cc_k^n \, \cc_2^n = \cc_k^n$ \cite{zeng2008semi}.

Gates $G \in \cc_3^n \setminus \cc_2^n$ in the strict third level of the Clifford hierarchy can be fault-tolerantly implemented via quantum gate teleportation to achieve universality when given access to an appropriate resource \emph{magic state} of $2n$ qudits.  The problem of implementing a non-Clifford gate is thus reduced to the problem of preparing a magic state; a task which can be done offline and in advance of a computation.  Higher-level gates can be implemented via a recursive procedure requiring additional ancillary resources.

\subsection{Semi-Clifford gates and one-bit teleportation}

Zhou-Leung-Chuang introduced a simplified gate teleportation protocol, based on Bennett-Gottesman's one-bit teleportation, capable of implementing certain qubit Clifford hierarchy gates using half the ancillary resources required in the original Gottesman-Chuang protocol.  This class of gates includes the diagonal Clifford hierarchy gates.  Zeng-Chen-Chuang introduced the notion of semi-Clifford gates which are `nearly diagonal' in the sense of being within Clifford corrections of diagonal Clifford hierarchy gates:

\begin{defn}\label{sc}A gate $G \in \cu(d^n)$ is \emph{semi-Clifford} if $G = C_1 D C_2$ where $C_1,C_2 \in \cc_2^n$ and $D$ is diagonal.\end{defn}

\begin{defn}For $k \geq 1$ the $k$-th level semi-Clifford gates are: $$\semc_k^n = \{G \in \cc_k^n \;|\; G = C_1 D C_2 \text{ where } C_1,C_2 \in \cc_2^n, D \in \cc_k^n \text{ is diagonal}\}.$$
\end{defn} \noindent They gave a protocol that expands the class of gates that can be implemented efficiently via one-bit gate teleportation to include the semi-Clifford gates.  We generalise this protocol to the qudit case in Section \ref{quditteleport}.

\section{The discrete Stone-von Neumann theorem}\label{discSVN}

\begin{defn}
An ordered pair of unitaries $(U,V) \in \cu(d^n) \times \cu(d^n)$ is a \emph{conjugate pair} if
\begin{enumerate}\itemsep0em 
\item $U^d = \I$ and $V^d = \I$,
\item $UV = \omega VU$.
\end{enumerate}
\end{defn}

\subsection{Single qudits}

\begin{lemma}\label{uvtraceless}
Suppose $U, V \in \cu(d)$ satisfy $UV = \omega VU$.  Then $U^p V^q$ are traceless for $(p,q) \in \Z_d^2$ with $(p,q) \neq (0,0)$.
\end{lemma}

\begin{proof}  Suppose first that $q \neq 0$.
$$ \mathrm{Tr}(U^p V^q) = \mathrm{Tr}(U^{p-1}U V^q) = \omega^q \mathrm{Tr}(U^{p-1} V^q U) = \omega^q \mathrm{Tr}(U^p V^q)$$
Since, $\omega^q \neq 1$, the above expression vanishes.  It similarly vanishes if $p \neq 0$.
\end{proof}

\begin{lemma}\label{uvbasis}
Suppose $U, V \in GL_d(\C)$ form a conjugate pair.  Then the matrices $\{U^p V^q \;|\; (p,q) \in \Z^2_d\}$ are orthogonal in $M_{d}(\C)$ with the Hilbert-Schmidt inner product $\langle A,B\rangle_{\mathrm{HS}} = \mathrm{Tr}(A^*B)$ and hence form a basis of $M_{d}(\C)$.
\end{lemma}

\begin{proof}  $\langle U^{p_1} V^{q_1},U^{p_2} V^{q_2}\rangle_{\mathrm{HS}} = \mathrm{Tr}(V^{-q_1} U^{-p_1} U^{p_2} V^{q_2}) \propto \mathrm{Tr}(U^{p_2-p_1} V^{q_2-q_1})$ as all terms can be commuted freely by possibly introducing factors of $\omega$ or $\omega^{-1}$.  This vanishes unless $p_1 \equiv p_2$ and  $q_1 \equiv q_2$ (mod $d$).    Since $U$, $V$ are unitary, their products are nonzero.  An orthogonal set of nonzero matrices is linearly independent.
\end{proof}

\begin{lemma}[Discrete Stone-von Neumann theorem, single-qudit version]\label{uniqueunitary}
Suppose $(U, V)$ and $ (\tilde{U},  \tilde{V})$ are two conjugate pairs.  There is a unitary $G$, unique up to phase, such that $\tilde{U} = GUG^*$ and $\tilde{V} = GVG^*$.
\end{lemma}

\begin{proof}
We define $\phi(U) =  \tilde{U}$ and  $\phi(V) =  \tilde{V}$ and prove that this extends to a unique $*$-automorphism of $M_{d}(\C)$.  The $*$-automorphisms of simple matrix algebras are in correspondence with unitaries up to phase (as a consequence of e.g. the Skolem-Noether theorem). 

We first define $\phi(U^i \, V^j) =  \tilde{U}^i \,  \tilde{V}^j$.  As the matrices $U^i \, V^j$ and $ \tilde{U}^i \,  \tilde{V}^j$  form two bases by Lemma \ref{uvbasis}, this is a well-defined vector space automorphism via its unique linear extension to all of $M_{d}(\C)$.  It is easy to check that it respects $*$ and matrix multiplication.  Thus, $\phi$ is an inner automorphism induced by a unitary $G$.
\end{proof}

\begin{thm}The unitary $G$ that carries $(Z,X)$ to $(U,V)$ under conjugation is given by: 
$$G \ket{z} = V^z \ket{u_0} $$
where $\ket{u_0} = G \ket{0}$ is an eigenvector of $U$ with eigenvalue 1.
\end{thm}

\begin{proof}
Apply $G$ to both sides of the equation: $\ket{z} = X^z \ket{0}$.
\end{proof}

\subsection{Multiple qudits}

Essentially the same proofs work to establish the multiqudit generalisations of the above results.

\begin{defn}
A \emph{conjugate tuple} $\{(U_1,V_2), ..., (U_n,V_n)\}$ is a set of $n$  conjugate pairs such that any two elements of distinct pairs commute.
\end{defn}

\begin{lemma}
Suppose $\{(U_i,V_i)\}_{i \in [n]}$ is a conjugate tuple.  Then the matrices $\{U^{\hat{p}} V^{\hat{q}} \;|\; (\hat{p},\hat{q}) \in \Z^{2n}_d\}$ are orthogonal in $M_{d^n}(\C)$ with the Hilbert-Schmidt inner product 
 and hence form a basis of $M_{d^n}(\C)$.
\end{lemma}

\begin{lemma}[Discrete Stone-von Neumann theorem]\label{uniqueunitarytuple}
Suppose $\{(U_i,V_i)\}_{i \in [n]}$  and $\{(\tilde{U}_i,\tilde{V}_i)\}_{i \in [n]}$ are two conjugate tuples.  There is a unitary $G$, unique up to phase, such that, for all $i \in [n]$, $\tilde{U_i} = GU_iG^*$ and $\tilde{V_i} = GV_iG^*$.
\end{lemma}

\begin{thm}\label{explicitunitary}The unitary $G$ that carries $\{(Z_i,X_i)\}_{i \in [n]}$ to $\{(U_i,V_i)\}_{i \in [n]}$ under conjugation is given by: 
$$G \ket{\hat{z}} = V^{\hat{z}} \ket{u_0} $$
where $\hat{z} \in \Z_d^n$ and $\ket{u_0} = G \ket{0}$ is a simultaneous eigenvector of the $U_1, ..., U_n$ with eigenvalue 1.
\end{thm}

The simplest way to construct $\ket{u_0}$ is to first compute the projectors $P_i = d^{-1} \sum_{p \in \Z_d} U_i^p$ onto the 1-eigenspaces of $U_i$ and to then compute an eigenvector $\ket{u_0}$ with eigenvalue 1 of the rank-1 projector $U_1 \cdots U_n$.

\subsection{An algorithm for enumerating the Clifford hierarchy gates}

\begin{defn}
A conjugate tuple is  \emph{$k$-closed} if it generates a group of $k$-th level gates: $$\{U^{\hat{p}} V^{\hat{q}} \;|\; (\hat{p},\hat{q}) \in \Z^{2n}_d\} \subset \cc^n_k.$$ 
\end{defn}

\begin{thm}\label{Ckfromtuples}
Gates of the $k+1$-th level of the Clifford hierarchy, up to phase, are in bijective correspondence with $k$-closed conjugate tuples.
\end{thm}

\begin{proof}
The correspondence is given by the map that sends $G \in [\cc_{k+1}^n]$ to the tuple $\{(U_i = GZ_iG^*, V_i = GX_iG^*)\}_{i \in [n]}$.  

Conjugation by $G$ preserves the order of all matrices and the commutation relations between them, so, as $\{(Z_i,X_i)\}_{i \in [n]}$ is a conjugate tuple, so is $\{(U_i,V_i)\}_{i \in [n]}$.  For any $(\hat{p},\hat{q}) \in \Z^{2n}_d$, the product $\prod_{i \in [n]} U_i^{p_i} V_i^{q_i} = G(\prod_{i \in [n]} Z_i^{p_i} X_i^{q_i})G^* \in \cc_k^n$ as it is the conjugation of a Pauli gate by a $\cc_{k+1}^n$ gate.

Conversely, given a $k$-closed conjugate tuple $\{(U_i, V_i)\}_{i \in [n]}$, we can apply Lemma \ref{uniqueunitarytuple} to find a unitary $G$, unique up to phase, such that $U_i = GZ_iG^*, V_i = GX_iG^*$.  The definition of $k$-closedness ensures that the conjugation of any Pauli gate by $G$ is in $\cc_k^n$ and thus, that, $G \in \cc_{k+1}^n$.
\end{proof}

Since the Clifford gates form a group, the condition of $2$-closedness is automatically fulfilled by any conjugate tuple of Cliffords.

\begin{thm}\label{C3fromtuples}
Gates of the third level of the Clifford hierarchy, up to phase, are in bijective correspondence with conjugate tuples of Clifford gates.
\end{thm}

The question of whether the assumption of $k$-closedness in Theorem \ref{Ckfromtuples} for $k > 2$ (those $k$ for which $\cc_k^n$ are not groups) is actually a necessary one remains open.  
%

Numerical investigations suggest that the assumption of $k$-closedness is not necessary for $(d,n,k)$ with $n = 1$.  We are therefore led to the following conjecture:

\begin{conj}
Gates of the $k+1$-th level of the Clifford hierarchy, up to phase, are in bijective correspondence with conjugate tuples of $k$-th level gates.
\end{conj}

An equivalent reformulation is:
\addtocounter{conj}{-1}
\begin{conj}\label{conjkclosed}
All conjugate tuples of $k$-th level gates are $k$-closed.
\end{conj}

That is, while $\cc_k^n$ may not be a group, does it contain a copy of the Heisenberg group whenever it contains its generators?  While the conjecture may not hold in general, it certainly holds for some triples $(d,n,k)$ and it is worth asking: for which ones?.

We can use Theorem \ref{Ckfromtuples} to describe a simple algorithm for recursively enumerating all gates of the Clifford hierarchy that works for any dimension and number of qudits.

\begin{algorithm}[H]
\floatname{algorithm}{Algorithm 1:  }
\renewcommand{\thealgorithm}{}
\caption{Recursively enumerate the $\cc_k^n$ gates (up to phase)}
\label{reck}

\begin{enumerate}
\item 
Generate $[\cc_1^n]$: the $d^{2n }$ Pauli gates without phase.

\item For $k = 2$ to $\infty$:

\begin{enumerate}
\item 
Select those elements of $[\cc_{k-1}^n]$ with order $d$ for some choice of phase
\item From all pairs of these elements, select the conjugate pairs
\item From $n$-tuples of conjugate pairs, select the conjugate tuples
\item From the conjugate tuples, select the $k$-closed conjugate tuples
\item From the $k$-closed conjugate tuples, generate the $[\cc_{k}^n]$ gates using Theorem \ref{Ckfromtuples}

\end{enumerate}

\end{enumerate}

\end{algorithm} 

In step 2-a, the elements of $[\cc_{k-1}^n]$ we are interested in are those for which raising any representative to the $d$-th power gives a diagonal matrix with constant diagonal element.  A choice of phase to correct that representative yielding one with order $d$ is easily extracted from this constant.  Note that the list of $k$-closed conjugate tuples generated by step 2-d is not complete.  That is because the conjugate pairs selected in step 2-b are selected from an enumeration of $\cc_{k-1}^n$ that ignores phase.  An ordered pair $(U,V)$ is a conjugate pair if and only if, for any $p,q \in \Z_d$, $(\omega^p U, \omega^q V)$ is a conjugate pair.  Therefore, in executing step 2-e, every $k$-closed conjugate tuple found by step 2-d generates $d^{2n}$ gates of $[\cc_{k}^n]$ by introducing an arbitrary choice of discrete phase factors into the elements of the conjugate tuple.

We computed $[\cc_{k}^1]$ in the $n=1$ case for small $d$ and $k$.  In the $d=3$ case, the sizes found were: 9, 216, 1944, 7128, 22680, 69336.  In the $d=5$ case, the sizes found were: 25, 3000, 7500, 435000, 2235000.  In the $d=7$ case, the sizes found were: 49, 16464, 806736, 6338640.

In order to check $k$-closure in step 2-d, one must implement a function to determine whether a gate is in $\cc^n_k$.  This can be defined recursively by conjugating all Pauli gates and checking if they are in $\cc^n_{k-1}$.

Deeper understanding of the structure of the Clifford hierarchy can lead to efficiency gains in the practical execution of this algorithm.  For example, establishing Conjecture \ref{conjkclosed} would eliminate the need for step 2-d.  A better grip on the lifting of the projective Weil representation to the ordinary representation could aid in optimising 2-b.

\section{Diagonal gates of the Clifford hierarchy}

In this section, we give an concise, elementary proof of Cui-Gottesman-Krishna's characterisation of diagonal gates of the qudit Clifford hierarchy in the single-qudit case.  Thus, $d$ is hereafter restricted to denoting an \emph{odd, prime} dimension.  For convenience, in this section, we drop superscripts $n$ indicating the number of qudits.

\subsection{Preliminaries definitions}
\begin{defn}For $k \geq 1$ let $$\cd_k = \{D \in \cc_k \;|\; D \text{ is diagonal and } D \ket{0} = \ket{0}\}.$$
\end{defn} \noindent The second condition ensures that $\cd_k$ contains precisely one gate up to a global phase factor.

Any integer $k \geq 1$ can be uniquely expressed as $$k = (m_k-1)(d-1) + a_k$$ with $a_k \in \{1, ..., d-1\}$.  We will suppress subscripts; call $m \geq 1$ the \emph{precision} of $k$ and $a$ the \emph{degree} of $k$.

\begin{defn}\label{rkdefn}Denote by $\crr_k$ the set of \emph{rank-$k$ polynomials}: \begin{align*}
\crr_k &= \Set{ \begin{array}{l} \phi: \Z_{d^m} \to \Z_{d^m} \smallskip\\  
\phi(z) = \sum_{j=1}^{d-1} \phi_j z^j \end{array} \ | \begin{array}{l}
    \phi \text{ has degree at most } d-1, \; \phi(0) = 0\smallskip\\
    \phi_{a+1}, ..., \phi_{d-1} \equiv 0 \;\;\mathrm{ (mod }\;d) \\
  \end{array}}.
\end{align*}
\end{defn}  \noindent Note that $\crr_k$ is an additive subgroup of the group of polynomials over $\Z_{d^m}$ and thus $$\crr_k \simeq \Z_{d^{m}}^a \times \Z_{d^{m-1}}^{(d-1)-a}.$$
Each copy of $\Z_{d^{m}}$ tracks the coefficient for the terms of degree $1, ..., a$ while each copy of $\Z_{d^{m-1}}$ tracks the coefficients of the terms of degree $a+1, ..., d-1$ after having divided out by a factor of $d$.  We can see immediately that $|\crr_k| = d^k$.

\begin{remark}We will also require the notion of a polynomial $\xi$ being of \emph{rank $k$ up to a constant}, i.e. there exists $C \in \Z_{d^m}$ such that $\xi + C \in \crr_k$.  If $\phi \in \crr_k$ and $\psi: \Z_{d^m} \to \Z_{d^m}$ is a polynomial of degree at most $d-1$, then $\phi + d\psi$ is rank $k$ up to a constant with $C = -d\psi(0)$.
\end{remark}

For an integer $m \geq 1$, denote the $d^m$-th primitive root of unity: $$\omega_m = e^{\imi \,  2 \pi / d^m}.$$

We now define the sets of gates that we will prove are the diagonal $k$-th level gates.

\begin{defn}$\Delta_k = \{D[{\omega_m}^\phi] \;|\; \phi \in \crr_k \}.$
\end{defn}

\begin{remark}
Recall that we defined the construction of diagonal gates $D[f]$ for $f: \Z_d \to \C$ whereas the polynomials $\phi$ in the definition above take elements of $\Z_{d^m}$ as their inputs.  Thus, when interpreting $\phi$ acting on $\Z_d$, we are implicitly precomposing with the natural inclusion $\Z_d \xhookrightarrow{} \Z_{d^m}$.  

To appreciate the significance of this seemingly trivial point, consider the action of the translation operator $T: \C^{\Z_d} \to \C^{\Z_d}$ defined by $Tf(z) = f(z-1 \text{ (mod }d\text{)})$.  To apply $T^q$ to a polynomial $\phi$, and thus to ${\omega_m}^\phi$, we can simply substitute each instance of $z$ with $(T^q \mathbf{1})(z)$ where $\mathbf{1}(z) = z$.  If this expression is interpreted over $\Z_d$, this would be $z - q$.  However, as we are considering polynomials over $\Z_{d^m}$, we require a correction term in our substitution: $$z \mapsto (z-q) + d \, \chi_q(z)$$ where $\chi_q$ is the characteristic function of $[0,q)$.  We have thus established:
\end{remark}

\begin{lemma}\label{translate}
Suppose $\phi: \Z_{d^m} \to \Z_{d^m}$ and $q \in \Z_d$.  Then $T^q({\omega_m}^\phi)(z) = {\omega_m}^{\phi((z-q) + d \, \chi_q(z))}$.
\end{lemma}

\subsection{A simple proof of CGK's characterisation of $\cd_k$}

\begin{lemma}[Zeng-Chen-Chuang, 2008]\label{dxd}
Suppose $D_{k+1} \in \cd_{k+1}$.  There exist $D_{k} \in \cd_{k}$ and $\theta \in [0, 2 \pi)$ such that $$D_{k+1} X D_{k+1}^* = e^{\imi \, \theta} D_{k} X .$$
\end{lemma}

\begin{proof}
Let $D_{k} = D_{k+1} X D_{k+1}^*X^*$.  It is diagonal as it commutes with $Z$:  the phase acquired as a $Z$ passes through $X$ is cancelled by the one acquired as it passes through $X^*$.  Thus, $D_{k} Z D_{k}^* = Z \in \cc_{k}$.  Further, $D_{k} X D_{k}^* = (D_{k+1} X D_{k+1}^*) X (D_{k+1} X D_{k+1})^* \in \cc_{k}$.  $D_{k}$ can be corrected by a phase factor to ensure that $D_{k} \ket{0} = \ket{0}$. 
\end{proof}

\begin{thm}[Cui-Gottesman-Krishna, 2016]\label{cgk}  $\cd_k = \Delta_k$ for all $k \in \Z^+$.
\end{thm}

\begin{proof}

We proceed by induction on $k$.  The $k = 1$ case is straightforward as $$\cd_1 = \{Z^p \;|\; p \in \Z_d\} = \{D[\omega^{\phi_p}] \;|\; p \in \Z_d, \phi_p(z) = pz\} = \Delta_1.$$

So, let us assume that  $\cd_{k} = \Delta_{k}$ and prove that $\cd_{k+1} = \Delta_{k+1}$.  This will require two steps.  First, we will count the elements of $\cd_{k+1}$ and find that $|\cd_{k+1}| = |\Delta_{k+1}|$.  Then, we will show that $\Delta_{k+1} \subset \cd_{k+1}$.  As  $|\cd_{k+1}|$ and $|\Delta_{k+1}|$ are finite sets, this will establish their equality and complete our proof.

\smallskip
\noindent\emph{Step 1:   $|\cd_{k+1}| = |\Delta_{k+1}|$.}  

Recall that $D_{k+1} \in \cd_{k+1}$ is determined by its conjugate pair and, by the preceding lemma, $$(D_{k+1}ZD_{k+1}^*,D_{k+1}XD_{k+1}^*) = (Z, e^{\imi  \theta}D_{k} X).$$

There are $d^{k}$ possible choices for $D_{k}$ and as we shall now show, for each one, $d$ possible choices of $\theta$ such that $e^{\imi  \theta}D_{k} X$ has order $d$.  Suppose $D_{k} =  D[{\omega_m}^\phi]$.  By repeatedly applying Lemma \ref{xdiagcomm}, $$(D_{k} X)^d = (\prod_{j=0}^{d-1} D[T^j ({\omega_m}^\phi)]) X^d = D[{\omega_m}^{{\sum_{j \in \Z_d}} \phi(j)}].$$ Thus, $(e^{\imi  \theta} D_{k} X)^d = \I$ for precisely $$e^{\imi \, \theta} = e^{\imi  \theta_\alpha} := \omega_m^{-\bar{\phi}}\omega^{\alpha}$$ with $\bar{\phi} = \frac{1}{d}\sum_{j \in \Z_d} \phi(j)$, the average value of $\phi$ over $\Z_d$, and for any $\alpha \in \Z_d$.    

Each choice of $(D_{k}, \theta_\alpha)$ yields a distinct conjugate pair with $Z$ and hence an element of $\cd_{k+1}$.  This will follow from Theorem \ref{Ckfromtuples} once we establish that $(Z, e^{\imi  \theta_\alpha}D_{k} X)$ is $k$-closed: i.e. $Z^p (e^{\imi  \theta_\alpha}D_{k} X)^q \in \cc_k$ for all $(p,q) \in \Z_d^2$.  As $$Z^p \, (e^{\imi  \theta_\alpha}D_{k} X)^q = e^{\imi  q \theta_\alpha} \, Z^p \, (\prod_{j=0}^{q-1} D[T^j ({\omega_m}^\phi)]) \, X^q$$ and given that $\cc_k$ is closed under multiplication by Pauli gates and phase factors, it is sufficient to show that $\prod_{j=0}^{q-1} D[T^j ({\omega_m}^\phi)] \in \cc_k$.  Each term is the group $\{G \in \cc_k \, | \, G\text{ is diagonal}\}$  as $D[T^j ({\omega_m}^\phi)] = X^j D[{\omega_m}^\phi] X^{-j} \in \cc_k$ and our conclusion follows.

Thus, $|\cd_{k+1}| = d^k \cdot d = d^{k+1} =  |\crr_{k+1}| = |\Delta_{k+1}|$.

\smallskip
\noindent\emph{Step 2: $\Delta_{k+1} \subset \cd_{k+1}$. }

For this step, $m,a$ denote the precision and degree of $k+1$, not that of $k$.

Suppose $D[\omega_m^\phi] \in \Delta_{k+1}$ with $\phi \in \crr_{k+1}$, i.e.  $\phi(z) = \sum_{j=1}^{d-1} \phi_j z^j$ with $\phi_{a+1}, ..., \phi_{p-1} \equiv 0 \text{ (mod }d)$.    We will show that $D[\omega_m^\phi] \in \cd_{k+1}$.  It is sufficient to show that $D[{\omega_m}^\phi] ( X^{q} Z^{p} ) D[{\omega_m}^{-\phi}] \in \cc_k$ for $p, q \in \Z_d$ as every Pauli gate is of the form $X^{q} Z^{p}$ up to phase. As $$D[{\omega_m}^\phi] ( X^{q} Z^{p} ) D[{\omega_m}^{-\phi}] = D[{\omega_m}^\phi]  X^{q}   D[{\omega_m}^{-\phi}] Z^{p} = D[{\omega_m}^\phi]  D[T^q({\omega_m}^{-\phi})] ( X^{q} Z^{p} ),$$ it is sufficient to prove that  $D[{\omega_m}^\phi \cdot T^q({\omega_m}^{-\phi})] \in \cc_k$ for any $q \in \Z_d$.

By Lemma \ref{translate}, $D[{\omega_m}^\xi] = D[{\omega_m}^\phi \cdot T^q({\omega_m}^{-\phi})]$ with 
$$\xi(z) = \sum_{j=1}^{d-1} \, \phi_j \, (z^j - [(z-q) + d \, \chi_q(z)]^j) = \sum_{j=1}^{d-1} \, \phi_j \, [z^j - \sum_{\beta=0}^{j} \binom{j}{\beta} (z-q)^{j-\beta} d^\beta \, {\chi_q(z)}^\beta ] .$$
We now separate out the $\beta=1$ and higher terms of the latter inner sum and divide by their common factor of $d$.  Noting that a polynomial with $d$ prescribed values can be constructed with degree at most $d-1$, let $\psi: \Z_{d^m} \to \Z_{d^m}$ be a polynomial of degree at most $d-1$ that,  on inputs in $\Z_d \subset \Z_{d^m}$, coincides with the resulting expression:
$$\psi|_{\Z_d}(z) = \sum_{j=1}^{d-1} \, \phi_j \, [\sum_{\beta=1}^{j}  \binom{j}{\beta} (z-q)^{j-\beta} d^{\beta-1} \, {\chi_q(z)}^\beta].$$

We define
$$\xi'(z) = \sum_{j=1}^{d-1} \, \phi_j \, [z^j - (z-q)^j] \; - d \, \psi(z)$$
and note that $\xi'|_{\Z_d}(z) = \xi|_{\Z_d}(z)$ and thus that $D[{\omega_m}^{\xi'}] = D[{\omega_m}^\xi]$.  It is therefore sufficient for us to prove that $\xi'$ is a rank-$k$ polynomial up to a constant.  By the Remark following Definition \ref{rkdefn}, we may add $d \psi$ to $\xi'$ whilst preserving its rank up to a constant.  We may similarly drop all terms for $j > a$ as these $\phi_j \equiv 0 \text{ mod }d$.  The coefficient for the $z^a$ term in the resulting expression,
$$\sum_{j=1}^{a} \, \phi_j \, (z^j - (z-q)^j),$$
vanishes, and so what remains has rank $k$ up to a constant.  This implies that there exists $C \in \Z_{d^m}$ such that $\xi' + C \in \crr_k$ and, thus, $D[{\omega_m}^\phi \cdot T^q({\omega_m}^{-\phi})] = D[{\omega_m}^{\xi'}] \in \cc_k$ as ${\omega_m}^{-C} D[{\omega_m}^{\xi'}] \in \cd_k$.
\end{proof}


\section{Semi-Clifford operators and efficient gate teleportation}\label{semisclifftelesec}

\subsection{Efficient qudit gate teleportation}\label{quditteleport}

We will now construct a circuit gadget that implements a semi-Clifford third level gate using half the ancillary quantum resources as required in the original gate teleportation protocol due to Gottesman-Chuang.  It is a generalisation to the qudit case of the qubit circuit due to Zhou-Leung-Chuang \cite{zhou2000methodology}.

Suppose $G \in \semc_3$, i.e. $G = C_1 D C_2$ for $C_1,C_2 \in \cc_2$, $D \in \cd_3$.  If given access to a magic state $\ket{M} = D \ket{+}$ we can perform $G$ on an input state $\ket{\psi}$ with the following circuit:

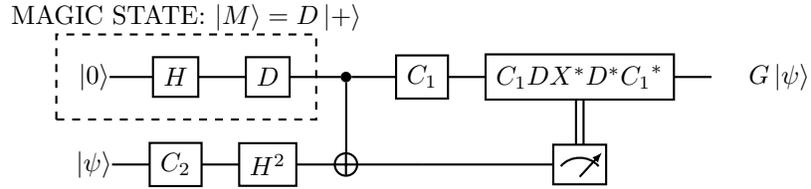
\begin{center}

\begin{quantikz}
\ket{0}\gategroup[wires=1,steps=3, style={dashed}]{{MAGIC STATE: $\ket{M} = D\ket{+}$}}    \& \gate{H} \& \gate{D} \& \ctrl{1} \& \gate{C_1}   \& \gate{C_1 D X^* D^* {C_1}^*} \& \qw \& G\ket{\psi} \\
\ket{\psi} \& \gate{C_2} \& \gate{H^2} \& \targ{} \& \qw    \& \meter{}  \vcw{-1}
\end{quantikz}

\end{center}

\noindent Crucially, the elements of the gadget are Clifford operations meaning that they can be implemented fault-tolerantly.  Preparing the magic state can be done fault-tolerantly and with greater efficiency than in the qubit case \cite{quditmsd}.

There are three key differences that manifest only in the qudit case.  First, the need for the $H^2$ gate with the action $H^2 \ket{z} = \ket{-z}$ which is simply the identity in the qubit case.  Second, the qubit $\textsc{CNOT}$ gate is generalised to the $\textsc{CX}$ (alternatively, $\textsc{CSUM}$) gate with the action $\textsc{CX} \ket{z_1}\ket{z_2} = \ket{z_1}\ket{z_1 + z_2 \text{ (mod }d)}$.  Finally, the need for $X^*$ in the final gate which is simply $X$ in the qubit case.

The validity of this circuit is most easily demonstrated by first considering the single-qudit $X$-teleportation circuit:

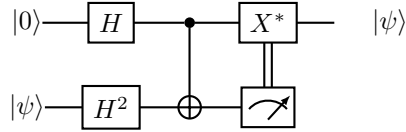
\begin{center}

\begin{quantikz}
\ket{0}    \& \gate{H} \& \ctrl{1}   \& \gate{X^*} \& \qw \& \ket{\psi} \\
\ket{\psi} \& \gate{H^2} \& \targ{}    \& \meter{}  \vcw{-1}
\end{quantikz}

\end{center}
\noindent Consider the action of this circuit on an input state $\psi = \sum_j \psi_j\ket{0}\ket{j}$.  It is mapped by $H \otimes H^2$ to $\frac{1}{\sqrt{p}}\sum_{i,j} \psi_j \ket{i}\ket{-j}$ which is then mapped by $\textsc{CX}$ to $\frac{1}{\sqrt{p}}\sum_{i,j} \psi_j \ket{i}\ket{i-j}$.  A measurement outcome of $J$ on the second qudit collapses the state to $\sum_{j} \psi_j \ket{j+J}\ket{J}$.  Applying the classically-controlled correction $X^{-J}$ to the first qudit and discarding the second qudit yields $\ket{\psi}$.

For any gate $G$ that commutes with $\textsc{CX}$, we can apply $G$ at the end of the circuit to yield an output of $G\ket{\psi}$ and commute the gate $G$ backwards in time until it is absorbed into the stage of preparation of a magic state:

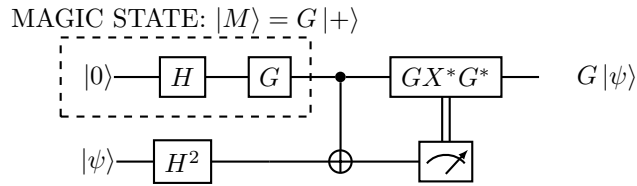
\begin{center}

\begin{quantikz}
\ket{0}\gategroup[wires=1,steps=3, style={dashed}]{{MAGIC STATE: $\ket{M} = G\ket{+}$}}    \& \gate{H} \& \gate{G} \& \ctrl{1}   \& \gate{GX^*G^*} \& \qw \& G\ket{\psi} \\
\ket{\psi} \& \gate{H^2} \& \qw \& \targ{}    \& \meter{}  \vcw{-1}
\end{quantikz}

\end{center}
\noindent This construction is particularly useful for $G \in \cd_3$ as $GX^*G^*$ is guaranteed to be Clifford and diagonal gates commute with the control of a controlled gate.  From this, we can generalise to implement $G \in \semc_3$ and obtain our first circuit: teleport the state $C_2 \ket{\psi}$ and apply $C_1D$ at the end of the circuit.

The above arguments are straightforwardly parallelised for $n$-qudit gates.  Zhou-Leung-Chuang show, in the qubit case, how a recursive construction can implement higher-level semi-Clifford gates with savings on ancillary resources.

\subsection{All $\cc_3^1$ gates admit efficient gate teleportation}

Before proving our main theorem, we will give an alternative characterisation of semi-Clifford gates in terms of their action on elementary Pauli gates.  This characterisation is a mild strengthening of Proposition 5 of \citep{zeng2008semi} that enables more efficient computations and simpler analytic proofs.  Note that Zeng-Chen-Chuang's proofs that, in the qubit case, $\cc_k^1 = \semc^1_k$ and  $\cc_k^2 = \semc^2_k$ make use of exhaustive computations.

\begin{defn}A \emph{Lagrangian semibasis} of a symplectic vector space of dimension $2n$ is a linearly independent set of $n$ vectors $\{v_1, ..., v_n\}$ satisfying $[v_i, v_j] = 0$ for all $i,j \in [n]$.

\end{defn}

\begin{lemma}
For any Lagrangian semibasis $\{(\hat{p_i},\hat{q_i})\}_{i \in [n]} = \{(p_1, q_1), ..., (p_n, q_n)\} \subset \Z_d^{2n}$, there is a Clifford $C \in \cc^n_2$ such that $CZ_iC^* = W(\hat{p_i} ,\hat{q_i} )$ for all $i \in [n]$.
\end{lemma}

\begin{proof}

A Lagrangian semibasis can be extended to a symplectic basis of $\Z_d^{2n}$.  The Clifford $C$ arising from the symplectic transformation that maps the standard basis to this symplectic basis yields the desired action.
\end{proof}

\begin{lemma}\label{customcliff}
For any set  $\{P_1, ..., P_n\}$ of $n$ independent, commuting Pauli gates, there is a Clifford $C \in \cc^n_2$ such that $CZ_iC^* = P_i$ for all $i \in [n]$.
\end{lemma}

\begin{proof}There exists a Lagrangian semibasis $\{(\hat{p_i},\hat{q_i})\}_{i \in [n]}$ and  $x_i \in \Z_d$ such that each Pauli $P_i = \omega^{x_i} W(\hat{p_i} ,\hat{q_i} )$.  Find $C \in \cc_2$ from the previous lemma such that $CZ_iC^* = W(\hat{p_i} ,\hat{q_i} )$.  Then, $C' = C {X_1}^{x_1} \cdots {X_n}^{x_n}$ yields the desired action.
\end{proof}

\begin{thm}\label{semicliff}
Suppose $G \in \cc_k^n$ and denote by $U_i = G{Z_i}G^*, V_i = G{X_i}G^*$ elements of $\cc_{k-1}^{n}$.  $G$ is semi-Clifford if and only if there exists a Lagrangian semibasis $\{(\hat{p_i}, \hat{q_i})\}_{i \in [n]} \subset \Z^{2n}_d$ such that, for each $i \in [n]$, $U^{\hat{p_i}} V^{\hat{q_i}}$ is a Pauli gate.
\end{thm}

\begin{proof}For notational simplicity, $\cc_k^n, \cd_k^n, \semc_k^n$ shall be denoted by $\cc_k, \cd_k, \semc_k$.
\begin{align*}
G \in \semc_k &\overset{(1)}{\iff} \exists\, C_1, C_2 \in \cc_2, D \in \cd_k \text { s.t. }G = C_1 D C_2 \\
& \overset{(2)}{\iff} \exists\,  C_1, C_2 \in \cc_2  \text { s.t. }C_1 G C_2 \in \cd_k \\
& \overset{(3)}{\iff} \exists\,  C_1, C_2 \in \cc_2  \text { s.t. } \forall\, i \in [n], C_1 G C_2 Z_i C_2^* G^* C_1^* = Z_i  \\
& \overset{(4)}{\iff} \exists\,  C_2 \in \cc_2  \text { s.t. } \{G C_2 Z_i C_2^* G^* \}_{i \in [n]} \text{ are $n$ independent, commuting Pauli gates}  \\
& \overset{(5)}{\iff} \exists \text{ a Lagrangian semibasis }  \{(\hat{p_i}, \hat{q_i})\}_{i \in [n]} \subset \Z^{2n}_p  \text{ s.t. } \forall\, i \in [n], U^{\hat{p_i}} V^{\hat{q_i}} \text{ is a Pauli gate.}
\end{align*}

The first two equivalences are straightforward and the third equivalence follows from Lemma \ref{diagmatcommz}: the fact that a matrix is diagonal if and only if it commutes with $Z_i$ for all $i \in [n]$.

For the direction $\overset{(4)}{\implies}$, we note that as the $Z_i$ are independent and commuting Paulis, so are $C_1^* Z_i C_1$.  The converse follows by applying Lemma \ref{customcliff} to construct $C_1^*$ and hence $C_1$.

For the direction $\overset{(5)}{\implies}$, the Lagrangian semibasis arises from the images $C_2Z_iC_2^* = \omega^{x_i} W(\hat{p_i} ,\hat{q_i} )$.  Conversely, we again apply Lemma \ref{customcliff} to construct $C_2$.\end{proof}

In addition to providing a useful technical characterisation, employed below, it can be used to algorithmically find the Cliffords that diagonalise a given semi-Clifford gate in a relatively efficient manner.

\begin{algorithm}[H]
\floatname{algorithm}{Algorithm 2:  }
\renewcommand{\thealgorithm}{}
\caption{Recognise and diagonalise semi-Clifford gates}
\label{recdiagsc}

\begin{enumerate}
\item Check if the given gate $G$ is in $\cc_k^n$; terminate if it is not.  Otherwise, store $U_i = GZ_iG^*$, $V_i = GX_iG^*$.

\item 
For each Lagrangian subspace $L$ of $\Z^{2n}_d$:

\begin{enumerate}
\item 
Choose a Lagrangian semibasis  $\{(\hat{p_i}, \hat{q_i})\}_{i \in [n]}$ of $L$

\item 
For each $i \in [n]$:
\begin{enumerate}
\item Check if $U^{\hat{p_i}} V^{\hat{q_i}}$ is a Pauli gate; return to step 2 if it is not.  
\end{enumerate}

\item If a Lagrangian semibasis satisfies the criterion of Theorem \ref{semicliff} store it and go to step 4.

\end{enumerate}

\item If no Lagrangian semibasis satisfies the criterion of Theorem \ref{semicliff}, terminate.

\item Construct $C_2$ as a Clifford satisfying $C_2 Z_i C_2^* = Z^{\hat{p_i}} X^{\hat{q_i}}$ using Lemma \ref{customcliff}.

\item Construct $C_1^*$ as a Clifford satisfying $C_1^* Z_i C_1 = U^{\hat{p_i}} V^{\hat{q_i}}$ using Lemma \ref{customcliff}.

\item Return $C_1$ and $C_2$.

\end{enumerate}

\end{algorithm} 

One simple way of generating a list of Lagrangian semibases, one for each Lagrangian subspace, is to first construct the list of vectors in $\Z_d^{2n}$ with leading nonzero component equal to 1 and to select from this list those subsets of size $n$ with pairwise vanishing symplectic product.

\begin{thm}\label{3rdlvlisSC}
Every third level gate of one qudit (of any prime dimension) is semi-Clifford: $\semc_3^1 = \cc_3^1$.
\end{thm}

\begin{proof}
Suppose $G \in \cc_3^1$ is a one-qudit third level gate and let $U = GZG^*$, $V = GXG^*$ be its conjugate pair of Cliffords.  Let $(S,v), (T,w) \in Sp(1, \Z_d) \ltimes \Z_d^{2}$ be such that $$\mu(S,v) = [U] \text{ and } \mu(T,w) = [V].$$  As $U V = \omega V U$, it follows that $[U][V] = [V][U]$ and thus that $S, T$ commute in $Sp(1, \Z_d)$.  We can thus define a group homomorphism: $$\phi_G: \Z^2_d \to Sp(1, \Z_d) := (p,q) \mapsto S^pT^q.$$  The order of $Sp(1, \Z_d)$ is $d(d^2 - 1)$ and, as this is not divisible by $|\Z_d^2| = d^2$, $\phi_G$ cannot be injective.  Therefore, there exists a nonzero vector $(p,q)$, i.e. a Lagrangian semibasis for $\Z^2_d$, such that $S^pT^q$ is the identity from which we can conclude that $U^pV^q$ is a Pauli gate.  By the previous theorem, $G \in \semc_3^1$.
\end{proof}

Numerical evidence supports the following conjecture:

\begin{conj}
Every $k$-th level gate of one qudit (of any prime dimension) is semi-Clifford: $\semc_k^1 = \cc_k^1$.
\end{conj}

\subsection{All two-qutrit $\cc_3^2$ gates admit efficient gate teleportation}

We now show that all third level gates of two-qutrits are semi-Clifford.  First, we establish general lemmas enabling reduction to the case of gates whose conjugate tuples consist of Cliffords of the form $D X^{\hat{\alpha}}$ for $D \in \cd_3^2$, $\hat{\alpha} \in \Z_d^2$ up to phase.

Denote by $Q \subset Sp(n,\Z_d)$ the abelian subgroup of symplectic matrices of the form $\begin{pmatrix}
\I & 0 \\
b & \I \end{pmatrix}$ for $n \times n$ symmetric matrices $b$.  Note that by counting the symmetric matrices, $|Q| = d^{n(n+1)/2}$.  Under the explicit representation of $Sp(n,\Z_d)$ of \citep{neuhauser2002explicit}, the image of $Q$ is a subgroup of diagonal Cliffords; each one is, up to phase, of the form $D[\omega^\phi]$ for $\phi: \Z_d^n \to \Z_d$ a homogeneous quadratic polynomial.
\begin{lemma}
Suppose that $B \subset Sp(n,\Z_d)$ is a set of commuting symplectic matrices such that $S^d = \I$ for all $S \in B$.  There exists $R \in Sp(n,\Z_d)$ such that $R\langle B \rangle R^{-1} \subset Q$.
\end{lemma}

\begin{proof}
The subgroup $\langle B \rangle$ generated by $B$ is either trivial (in which case, the lemma follows immediately) or it is a $d$-subgroup of $Sp(n,\Z_d)$.  Therefore, $\langle B \rangle \subset A \subset J$ where $J$ is a Sylow $d$-subgroup of $Sp(n,\Z_d)$ and $A$ is a maximal abelian subgroup of $G$.  

By Theorem 2.5 of \cite{barry1979large}, the Sylow $d$-subgroups of $Sp(n,\Z_d)$ contain a \emph{unique} maximal abelian subgroup of order $d^{n(n+1)/2}$.  Note that non-identity every element of $Q$ has order $d$ and so $Q$ is contained in a Sylow $d$-subgroup $K$ of$Sp(n,\Z_d)$.  As all Sylow $d$-subgroups are conjugate, there exists $R \in Sp(n,\Z_d)$ such that $RJR^{-1} = K$.  

As $Q$ has order $d^{n(n+1)/2}$, it is the maximal abelian subgroup of $K$.  Conjugation by $R$ must carry the unique maximal abelian subgroup of $J$ to that of $K$ and so $RAR^{-1} \subset Q$ implying the lemma.
\end{proof}

\begin{lemma}\label{diagsymp}
Suppose that $G \in \cc_3^n$ is a third level gate with the corresponding  conjugate tuple $\{(U_i,V_i)\}_{i \in [n]}$ of Clifford gates.  Suppose further that  $(S_i,v_i), (T_i,w_i) \in Sp(n,\Z_d) \ltimes \Z_d^{2n}$ are such that $\rho(S_i,v_i) = [U_i]$ and $\rho(T_i,w_i) = [V_i]$ as elements of $[\cc_2^n]$.  There exists a Clifford gate $C \in \cc_2^n$ such that $CGC^*$ has conjugate tuple $\{(\omega^{x_i} D_i X^{\hat{\alpha_i}},\omega^{y_i}E_i  X^{\hat{\beta_i}})\}_{i \in [n]}$ where $D_i, E_i \in \cd_3^n$; $x_i,y_i \in \Z_d$; $\hat{\alpha_i}, \hat{\beta_i} \in \Z_d^n$.
\end{lemma}

\begin{proof}
Apply the preceeding lemma to $E = \{S_1, T_1, ..., S_n, T_n\}$ and take $C = \rho(R, 0)$ (with any phase).
\end{proof}

The situation where the conjugate tuples are diagonal Cliffords multiplied by a Pauli is easier to tackle as we can characterise them in a simpler combinatorial fashion.

\begin{lemma}\label{c3eqns}
Suppose $U = \omega^{x} D X^{\hat{\alpha}}$, $V = \omega^{y} E X^{\hat{\beta}}$, for $D, E \in \cd_3^2$; $x,y \in \Z_d$; and $\hat{\alpha}, \hat{\beta} \in \Z_d^2$.  Suppose further that $D = D[\omega^\phi] Z^{\hat{a}}$ and $E = E[\omega^\psi] Z^{\hat{b}}$ where $\phi(\hat{z}) = d_1 z_1^2 + d_2 z_2^2 + d_3 z_1 z_2$; $\psi(\hat{z}) = e_1 z_1^2 + e_2 z_2^2 + e_3 z_1 z_2$; and $\hat{a}, \hat{b} \in \Z_d^2$.

Then $UV = \omega^c VU$ for $c \in \Z_d$ if and only if, modulo $d$:
\begin{align*} 2 d_1 \beta_1 + d_3 \beta_2 -2 e_1 \alpha_1 - e_3 \alpha_2 &\equiv 0 \\
 d_3 \beta_1 + 2 d_2 \beta_2  -e_3 \alpha_1 - 2 e_2 \alpha_2  &\equiv 0 \\
 \psi(\hat{\alpha}) + \hat{a} \cdot \hat{\beta}  - \phi(\hat{\beta}) - \hat{b} \cdot \hat{\alpha}  &\equiv c.
\end{align*}
\end{lemma}

\begin{proof}Simplify the expression $UVU^*V^* = \omega^c \, \I$ by commuting $X_1, X_2$ terms to the right by repeatedly applying Lemma \ref{xdiagcomm}.  The effect of translating the polynomials defining $U$ and $V$ is to multiply by a Pauli correction; the above equations ensure that this correction is simply the desired phase factor $\omega^c$.
\end{proof}


These lemmas can be applied towards computationally verifying that each two-qutrit third level gate is semi-Clifford.  Computations would be intractable without them.  We then indicate how the proof might be analytically generalised to higher dimensions.

\begin{thm}
Every third level gate of two qutrits is semi-Clifford.
\end{thm}

\begin{proof}
Let $G$ be a two-qutrit gate of the third level of the Clifford hierarchy.  By Lemma \ref{diagsymp}, we can assume that without loss of generality that $G$ has a conjugate tuple of the form $\{(\omega^{x_1} D_1 X^{\hat{\alpha_1}},\omega^{y_1}E_1  X^{\hat{\beta_1}}),(\omega^{x_2} D_2 X^{\hat{\alpha_2}},\omega^{y_2}E_2  X^{\hat{\beta_2}})\}$ where $D_i, E_i \in \cd_3^2$; $x_i,y_i \in \Z_d$; and $\hat{\alpha_i}, \hat{\beta_i} \in \Z_d^2$.  We can ignore the discrete phases as doing so results in another conjugate tuple whose corresponding gate is, by Theorem \ref{semicliff}, semi-Clifford if and only if the original one is.  

Thus, $U = GZ_1G^*, V = GX_1 G^*, S = G Z_2 G^*, T = G X_2 G^*$ are characterised by four septuples of elements of $\Z_d$ that satisfy the eighteen equations of Lemma \ref{c3eqns} describing the commutation relations $UV = \omega VU$, $ST = \omega TS$, $US = SU$, $UT = TU$, $VS = SV$, $VT = TV$.  One can exhaustively compute all such quadruples of septuples by first computing the conjugate pairs and then by finding the pairs of these which give conjugate tuples; we find there to be 4199040 such conjugate tuples.

One can then apply Theorem \ref{semicliff} to verify that each conjugate tuple arises from a semi-Clifford gate.  It is sufficient to verify that the kernel of the matrix $$\begin{pmatrix}
u_1 & v_1 & s_1 & t_1 \\
u_2 & v_2 & s_2 & t_2 \\
u_3 & v_3 & s_3 & t_3 \end{pmatrix}$$ contains a Lagrangian semibasis, where $u_i, v_i, s_i, t_i$ are the coefficients of the homogeneous quadratic polynomial of $U,V,S,T$ respectively.  This is because for any Pauli $P_1$ and homogeneous quadratic $\phi$, there is a Pauli $P_2$ such that $P_1 D[\omega ^\phi] = D[\omega ^\phi] P_2$.
\end{proof}

One path to generalising this result to higher dimensions would be to analytically derive from the eighteen equations of Lemma \ref{c3eqns} characterising conjugate tuples a Lagrangian semibasis in the kernel of the above matrix.

\begin{conj}
Every $k$-th level gate of two qudits (of any prime dimension) is semi-Clifford: $\semc_k^n = \cc_k^n$ for $n = 1,2$.
\end{conj}

%
%
%
%
%
%
%

\section{Conclusions and open problems}

Understanding the structure of the Clifford hierarchy and the semi-Clifford gates, i.e. those admitting efficient implementation via the one-dit gate teleportation protocol described above, in the qudit case is essential for bolstering the viability of qudit fault-tolerant quantum computation.

We have developed a perspective on studying the qudit Clifford hierarchy via the discrete Stone-von Neumann theorem.  This focus on studying Clifford gates via their actions by conjugation on basic Pauli gates, first employed by Beigi-Shor in the qubit, third level case, is fruitfully extended to the widest possible generality.

Technically, this perspective enables a simple proof of Cui-Gottesman-Krishna's classification of diagonal Clifford hierarchy gates (in the single-qudit case) which raises the question: might it more easily admit generalisation to a classification of \emph{all} Clifford hierarchy gates?  It further enables a novel characterisation of semi-Clifford gates that serves as the basis for proving that all third level gates of one-qudit and two-qutrits are semi-Clifford.

These technical developments lead to simple algorithms for recursively enumerating all members of the Clifford hierarchy that works for any $(d,n,k)$ and for recognising and diagonalising semi-Clifford gates.

We have employed these algorithms to find numerical evidence that support a number of conjectures.  Establishing these conjectures promise to stimulate development of the stabiliser formalism and of the Clifford hierarchy as well as further bolster the viability of qudit fault-tolerant quantum computation.

\addtocounter{conj}{-3}
\begin{conj}
Gates of the $k+1$-th level of the Clifford hierarchy, up to phase, are in bijective correspondence with conjugate tuples of $k$-th level gates.
\end{conj}

\begin{conj}
Every $k$-th level gate of one qudit (of any prime dimension) is semi-Clifford: $\semc_k^1 = \cc_k^1$.
\end{conj}

\begin{conj}
Every third level gate of two qudits (of any prime dimension) is semi-Clifford: $\semc_3^2 = \cc_3^2$.
\end{conj}

In analogy with the result of Zeng-Chen-Chuang that all two-qubit gates of any level are semi-Clifford, one might even conjecture that every $k$-th level gate of two qudits (of any prime dimension) is semi-Clifford: $\semc_k^2 = \cc_k^2$.

While these conjectures may not hold for all $(d,n,k)$, they do hold for some such triples and almost certainly for some beyond what has already been established rigorously.  They may therefore be reformulated as questions: for which  $(d,n,k)$ do they hold?

A future direction of research is to give a complete classification of the Clifford hierarchy.  One might begin by classifying the third level gates which, by Theorem \ref{C3fromtuples}, correspond to conjugate tuples of Clifford gates.  Using the explicit metaplectic representation of Neuhauser \cite{neuhauser2002explicit} could be useful in characterising these tuples.  

\begin{acknowledgments}
We wish to thank Mark Howard and Michael J.J. Barry for helpful conversations and Richard Jozsa for encouragement and support during this project.

We acknowledge support from the QuantERA ERA-NET Cofund in Quantum Technologies implemented within the European Union's Horizon 2020 Programme (QuantAlgo project), and administered through the EPSRC grant no. EP/R043957/1: Quantum algorithms and applications.
\end{acknowledgments}

\bibliographystyle{abbrv}

{\footnotesize
\bibliography{cliff}}

\end{document}